\documentclass[a4paper, 10pt, onecolumn]{ieeeconf}  
\IEEEoverridecommandlockouts                              
\usepackage{hyperref}

\usepackage{graphicx}

\usepackage{amsmath}
\usepackage{amssymb}
\usepackage{amsfonts}
\usepackage{datetime}

\newtheorem{thm}{Theorem}

\DeclareMathAlphabet{\bit}{OML}{cmm}{b}{it}
\def\<{\leqslant}           
\def\>{\geqslant}           

\def\d{\partial}
\def\wh{\widehat}
\def\wt{\widetilde}

\def\Re{\mathrm{Re}}   

\def\mR{{\mathbb R}}    
\def\mC{\mathbb{C}}    

\def\Tr{\mathrm{Tr}}       
\def\sign{\mathrm{sign}}       
\def\rT{{\rm T}}        
\def\diam{\diamond}       

\def\bE{\mathbf{E}}    

\def\bra{{\langle}}
\def\ket{{\rangle}}

\def\re{{\rm e}}        
\def\rd{{\rm d}}        

\def\bJ{\mathbf{J}}

\def\br{\mathbf{r}}
\def\x{\times}
\def\ox{\otimes}

\def\fF{\mathfrak{F}}

\def\fH{\mathfrak{H}}

\def\cC{\mathcal{C}}

\def\cI{\mathcal{I}}

\def\Ups{\Upsilon}

\def\diag{\mathop{\rm diag}}    

\title{\LARGE \bf
A Karhunen-Loeve Expansion for
One-mode Open Quantum Harmonic Oscillators Using the Eigenbasis of the Two-point Commutator Kernel$^*$}

\author{Igor G. Vladimirov$^\dagger$,\qquad
Matthew R. James$^\dagger$,\qquad
Ian R. Petersen$^\dagger$
\thanks{$^*$This work is supported by the Air Force Office of Scientific Research (AFOSR) under agreement number FA2386-16-1-4065 and the Australian Research Council under grant DP180101805.}
\thanks{$^\dagger$Research School of Electrical, Energy and Materials Engineering, College of Engineering and Computer Science,
Australian National University, Canberra, Acton, ACT 2601,
Australia, {\tt
igor.g.vladimirov@gmail.com, matthew.james@anu.edu.au, i.r.petersen@gmail.com}.
}
}

\begin{document}

\maketitle
\thispagestyle{empty}

\begin{abstract}
This paper considers one-mode open quantum harmonic oscillators with a pair of conjugate position and momentum variables driven by vacuum bosonic fields according to a linear quantum stochastic differential equation. Such systems model cavity resonators in quantum optical experiments. Assuming that the quadratic Hamiltonian of the oscillator is specified by a positive definite energy matrix, we consider a modified version of the quantum Karhunen-Loeve expansion of the system variables proposed recently. The expansion employs eigenvalues and eigenfunctions of the two-point commutator kernel for linearly transformed system variables. We take advantage of the specific structure of this eigenbasis in the one-mode case (including its connection with the classical Ornstein-Uhlenbeck process). These  results are applied to computing quadratic-exponential cost functionals which provide robust performance criteria for risk-sensitive control of open quantum systems.
\end{abstract}


\section{INTRODUCTION}

Similarly to the Fourier series for square integrable functions, classical random processes with finite second moments over bounded time intervals admit the Karhunen-Loeve (KL) expansion \cite{GS_2004} (see also \cite{IT_2010}). This representation employs an orthonormal basis of eigenfunctions of the covariance kernel of the process with uncorrelated random coefficients (which are independent in the Gaussian case). The truncation of the resulting series provides a meshless approximation (rather than time discretization) of the underlying random process and is similar in this regard to the Ritz-Galerkin methods \cite{M_1950}.

This approach has recently been extended in \cite{VPJ_2019b} to quantum processes which describe the Heisenberg evolution of dynamic variables of open quantum harmonic oscillators (OQHOs), which constitute a building  block of linear quantum systems theory \cite{NY_2017,P_2017}.  In the framework of the  Hudson-Parthasarathy
calculus \cite{HP_1984,P_1992,P_2015}, such systems are governed by linear quantum stochastic differential equations (QSDEs) driven by quantum Wiener processes on a symmetric Fock space \cite{PS_1972}, which represent bosonic quantum fields (quantised electromagnetic radiation).  The dynamics of OQHOs, affected by their interaction with the external fields, are specified by a quadratic system Hamiltonian and linear system-field coupling operators.

The quantum Karhunen-Loeve (QKL) expansion, proposed in \cite{VPJ_2019b}, represents the system variables of a stable OQHO (with a Hurwitz dynamics matrix) over a bounded time interval by a series of eigenfunctions of the invariant two-point quantum covariance kernel. In contrast to the classical case, the coefficients of the QKL expansion are
organised as conjugate pairs of noncommuting quantum mechanical positions and momenta \cite{S_1994},
whose statistical properties are described in quantum probabilistic terms \cite{H_2018,M_1995} and do not reduce to classical joint probability distributions.
However, the eigenfunctions of the quantum covariance kernel are not always available in closed form. At the same time, the imaginary part of this kernel, which describes the two-point canonical commutation relations (CCRs) of the system variables, lends itself to a complete eigenanalysis for a class of one-mode OQHOs with a position-momentum pair of system variables  and multichannel input fields.

For this class of one-mode OQHOs
with positive definite energy matrices and a stability condition on the system-field coupling, the present paper develops a modified version of the QKL expansion which uses the eigenbasis associated with the commutator kernel. Modulo a symplectic transformation of the system variables, the eigenvalues and eigenfunctions of the commutator kernel reduce to those for the covariance kernel of a classical Ornstein-Uhlenbeck process \cite{KS_1991}. This leads to sinusoidal eigenfunctions as closed-form solutions of a boundary value problem for a linear second-order ODE. The coefficients of the modified QKL expansion  for the transformed system variables are again position-momentum pairs with interpair commutativity. Moreover,  they are in a Gaussian quantum state if the OQHO is driven by vacuum fields, and their cross-covariances lend themselves to explicit computation.

We then apply the modified QKL expansion of the system variables to computing a quadratic-exponential functional (QEF) \cite{VPJ_2018a} (see also \cite{B_1996}), which also involves symplectic techniques   \cite{VPJ_2018c}. In comparison with the recent results on this approach in \cite{VPJ_2019b}, the present paper takes advantage of the specific features of the one-mode case (including the choice of a more tractable eigenbasis and its relations to the classical Ornstein-Uhlenbeck process). The relevance of this problem is explained by the fact that the QEF is an alternative (though closely related \cite{VPJ_2019a}) version of the original quantum risk-sensitive cost \cite{J_2004,J_2005}. Its minimization in quantum control and filtering problems (by an appropriate choice of a controller or filter for a given quantum plant) improves conservativeness of the closed-loop system in the sense of large deviations of quantum trajectories \cite{VPJ_2018a} and robustness with respect to uncertainties in the system-field state \cite{VPJ_2018b} described
in terms of quantum relative entropy \cite{OW_2010,YB_2009}. In addition to being an important robust performance analysis problem, the QEF computation is of interest in its own right and also has deep connections with operator algebras \cite{AB_2018}, 
the characteristic (or moment-generating) functions for quadratic Hamiltonians \cite{PS_2015} and the quantum L\'{e}vy area \cite{CH_2013,H_2018}.

The paper is organised as follows.
Section~\ref{sec:sys}  specifies the class of one-mode OQHOs under consideration.
Section~\ref{sec:eig} considers the spectrum and eigenfunctions for the two-point commutator kernel of the system variables of the OQHO.
Section~\ref{sec:QKL} employs this eigenbasis for a modified quantum Karhunen-Loeve
expansion of the system variables and discusses the statistical properties of the QKL coefficients in the case of vacuum input fields.
Section~\ref{sec:QEF} applies the modified QKL representation to computing the QEF for the one-mode OQHO.
Section~\ref{sec:conc} provides concluding remarks.

\section{ONE-MODE OPEN QUANTUM HARMONIC OSCILLATORS}
\label{sec:sys}

We consider a one-mode open quantum harmonic oscillator (OQHO) endowed with a pair of conjugate position $q$ and momentum $p$ variables,  which are assembled into a vector
\begin{equation}
\label{Xqp}
    X(t) =
    {\begin{bmatrix}
        q(t)\\
        p(t)
    \end{bmatrix}}
\end{equation}
(vectors are organised as columns unless indicated otherwise)
and evolve in time $t\>0$. These system variables are time-varying self-adjoint operators, satisfying the Weyl canonical commutation relations (CCRs) \cite{F_1989} in the Heisenberg infinitesimal  form
\begin{equation}
\label{XCCR}
    [X(t),X(t)^\rT]
    =
    {\begin{bmatrix}
     [q(t), q(t)] & [q(t), p(t)]\\
     [p(t), q(t)] & [p(t), p(t)]
    \end{bmatrix}}
    =
    i \bJ
\end{equation}
for any $t\>0$, with $i:= \sqrt{-1}$ the imaginary unit, and $[\alpha,\beta]:= \alpha\beta - \beta\alpha$ the commutator of linear operators. Here, use is made of the matrix
\begin{equation}
\label{bJ}
\bJ: = {\begin{bmatrix}
        0 & 1\\
        -1 & 0
    \end{bmatrix}}
\end{equation}
which spans the subspace of 
antisymmetric matrices of order 2, with $-i\bJ = \sigma_2$ being the second of the Pauli matrices \cite{S_1994}:
\begin{equation}
\label{Pauli}
    \sigma_1
    :=
    {\begin{bmatrix}
      0 & 1\\
      1 & 0
    \end{bmatrix}},
    \qquad
    \sigma_2
    :=
    {\begin{bmatrix}
      0 & -i\\
      i & 0
    \end{bmatrix}},
    \qquad
    \sigma_3
    :=
    {\begin{bmatrix}
      1 & 0\\
      0 & -1
    \end{bmatrix}}.
\end{equation}
The CCRs (\ref{XCCR}) hold, for example, for the operator of multiplication by the position variable $q$ and the differential operator $p:= -i\d_q$, acting on the Schwartz space
\cite{V_2002}. The latter is dense in the Hilbert space of square integrable functions, which can be used as an initial system space $\fH_0$ (for the action of the initial system variables $q(0)$, $p(0)$).  Associated with (\ref{XCCR}), (\ref{bJ}) is the CCR matrix of the system variables:
\begin{equation}
\label{Theta}
  \Theta
  :=
  \frac{1}{2} \bJ.
\end{equation}
The OQHO interacts with external bosonic fields which are modelled by a multichannel quantum Wiener process $W:=(W_k)_{1\< k \< m}$, consisting of an even number $m$ of time-varying self-adjoint operators $W_1(t), \ldots, W_m(t)$ on a symmetric Fock space $\fF$ \cite{P_1992}. In accordance with its continuous tensor-product structure \cite{PS_1972}, $\fF$ is endowed with an increasing family of subspaces $\fF_t$, so that $W_k(t)$ acts effectively on $\fF_t$ for any $t\>0$ and $k=1, \ldots, m$, with $(\fF_t)_{t\> 0}$ playing the role of a filtration for $\fF$.  The component  quantum Wiener processes, which initially are the identity operator $W_k(0)= \cI_{\fF}$ on the Fock space,  satisfy the two-point CCRs
$    [W(s), W(t)^\rT]
     := ([W_j(s), W_k(t)])_{1\< j,k\< m}
     =
    2i(s\wedge t)J
$
for all $s,t\>0$, with $s\wedge t:= \min(s,t)$ for the sake of brevity.
Here,
\begin{equation}
\label{JJ}
    J:=  \bJ \ox I_{m/2}
\end{equation}
is an orthogonal real  antisymmetric matrix (so that $J^2 = -I_m$), with $\ox$ the Kronecker product, and $I_r$ the identity matrix of order $r$.
Due to the interaction of the OQHO with the external bosonic fields,  the system variables $q(t)$, $p(t)$ in (\ref{Xqp}) act on the tensor-product Hilbert space $\fH_t:= \fH_0 \ox \fF_t$ (a subspace of the system-field space $\fH:= \fH_0 \ox \fF$). As a particular case of the Hudson-Parthasarathy calculus \cite{HP_1984,P_1992}, the evolution of the system variables is modelled by a linear QSDE
\begin{equation}
\label{dX}
    \rd X = AX \rd t + B\rd W
\end{equation}
driven by the quantum Wiener process $W$ (the time arguments are omitted for brevity), whose equivalent integral form is
\begin{equation}
\label{Xsol}
    X(t)
    =
    \re^{(t-s)A} X(s) + \int_s^t \re^{(t-\tau)A} B \rd W(\tau),
    \qquad
    t\> s\> 0.
\end{equation}
In view of (\ref{XCCR})--(\ref{Theta}), the matrices $A \in \mR^{2\x 2}$, $B\in \mR^{2\x m}$ are parameterised as
\begin{equation}
\label{AB}
    A = \bJ (R + M^\rT JM),
     \qquad
     B = \bJ M^\rT
\end{equation}
in terms of the energy matrix $R = R^\rT:=(r_{jk})_{1\< j,k\< 2} \in \mR^{2\x 2}$ and the coupling matrix $M \in \mR^{m\x 2}$ which specify the system Hamiltonian
\begin{equation}
\label{H}
    H :=
    \frac{1}{2} X^\rT R X
    =
    \frac{1}{2}
    (r_{11}q^2 + r_{12}(qp+pq) + r_{22}p^2)
\end{equation}
and the vector $MX$ of $m$ system-field coupling operators. The representation (\ref{AB}) is closely related to the physical realizability (PR) condition \cite{JNP_2008}
\begin{equation}
\label{PR}
    A \Theta + \Theta A^\rT + BJB^\rT = 0
\end{equation}
for the preservation of the CCRs (\ref{XCCR}) in time, with the CCR matrix $\Theta $ given by (\ref{Theta}). Similarly to classical linear SDEs,  the asymptotic behaviour of the system variables governed by the QSDE (\ref{dX}) (such as the existence of and convergence to an invariant quantum state)  depends  on whether the matrix $A$ in (\ref{AB}) is Hurwitz. To this end, note that $M^\rT J M$ is a real antisymmetric matrix of order $2$ in view of (\ref{JJ}), and hence, it is representable as
\begin{equation}
\label{MJM}
  M^\rT J M = \mu \bJ
\end{equation}
in terms of the basis matrix $\bJ$ from (\ref{bJ}) for some scalar $\mu \in \mR$.

\begin{thm}
 \label{th:stab}
 Suppose the energy matrix  $R$ of the one-mode OQHO (\ref{dX}) is positive definite.
 Then the
 eigenvalues of the matrix $A$ in (\ref{AB}) are given by $-\mu\pm i\nu$, where $\mu$ is specified by (\ref{MJM}), and
  \begin{equation}
 \label{freq}
        \nu := \sqrt{\det R}.
\end{equation}
Moreover, in this case,
\begin{equation}
\label{AR}
    A = R^{-1/2}
    {\begin{bmatrix}
      -\mu & \nu\\
      -\nu & -\mu
    \end{bmatrix}}
    \sqrt{R},
\end{equation}
where $R^{-1/2}$ is the inverse of the real positive definite symmetric matrix square root $\sqrt{R}$. In particular, if $\mu>0$ (in addition to $R\succ 0$), then the matrix $A$ is Hurwitz.\hfill$\square$
\end{thm}
\begin{proof}
Since the matrix $\bJ$ in (\ref{bJ}) satisfies $\bJ^2 = -I_2$, then it follows from (\ref{MJM}) and the assumption  $R\succ 0$ that the matrix $A$ in (\ref{AB}) takes the form
\begin{align}
\nonumber
    A
    & =
    \bJ (R + \mu \bJ)
    = \bJ R - \mu I_2\\
\nonumber
    &=
    R^{-1/2}(\sqrt{R}\bJ\sqrt{R} - \mu I_2)\sqrt{R}\\
\label{AR1}
    & =
    R^{-1/2}(\nu \bJ  - \mu I_2)\sqrt{R}.
\end{align}
Here, use is also made of the quantity $\nu$ from (\ref{freq}) along with the identity $S\bJ S^\rT = \bJ \det S$ for any matrix $S \in \mC^{2\x 2}$ (whereby a real $(2\x 2)$-matrix is symplectic \cite{D_2006} if and only if it has unit determinant). The right-hand side of (\ref{AR1}) is identical to that of (\ref{AR}). Therefore, the matrix $A$ is related by a similarity transformation (and hence, is isospectral) to the matrix ${\small\begin{bmatrix}
      -\mu & \nu\\
      -\nu & -\mu
    \end{bmatrix}}$ whose eigenvalues are $-\mu \pm i\nu$, so that the fulfillment of $\mu >0$ (together with $R\succ 0$) indeed makes $A$ Hurwitz.
\end{proof}

Theorem~\ref{th:stab} implies (under its conditions) that the linear transformation
\begin{equation}
\label{Xnew}
    \wt{X}:=  S X,
    \qquad
    S:= \sqrt{\frac{1}{\nu} R}
\end{equation}
yields a vector $\wt{X}$ of new self-adjoint system variables with the same CCR matrix (\ref{Theta}):
$    [\wt{X},\wt{X}^\rT]
    =
    S [X,X^\rT] S^\rT = iS\bJ S^\rT = i\bJ
$.
Indeed, in view of (\ref{freq}), the matrix $S \in \mR^{2\x 2}$ in (\ref{Xnew}) satisfies $\det S = \frac{1}{\nu}\sqrt{\det R}=1$ and is therefore symplectic.
With respect   to the new system variables, the OQHO has a scalar energy matrix $\wt{R}:= \nu I_2$, which is obtained by representing the Hamiltonian $H$ in (\ref{H}) as
$
    H = \frac{1}{2} (\sqrt{R} X)^\rT \sqrt{R} X = \frac{1}{2} \nu \wt{X}^\rT \wt{X}
$.
Similarly, the corresponding new coupling matrix is $\wt{M} := MS^{-1}$ since $MX = \wt{M} \wt{X}$. It has the same parameter $\mu$ in (\ref{MJM}) due to $S$ being symplectic:
$
    \wt{M}^\rT J \wt{M}
     = S^{-\rT}M^\rT J MS^{-1}
     = \mu S^{-\rT}\bJ S^{-1} =-\mu (S\bJ S^\rT)^{-1} = \mu\bJ
$,
where $S^{-\rT}:= (S^{-1})^\rT$, and use is made of the property $\bJ^{-1} = -\bJ$.
Furthermore, the new system variables satisfy the QSDE
$    \rd \wt{X} = \wt{A}\wt{X} \rd t + \wt{B}\rd W
$,
where, in view of (\ref{AR}), (\ref{Xnew}), the matrices $\wt{A} \in \mR^{2\x 2}$, $\wt{B} \in \mR^{2\x m}$ are given by
\begin{equation}
\label{ABnew}
    \wt{A} = SAS^{-1} = \sqrt{R} AR^{-1/2}
    =
    {\begin{bmatrix}
      -\mu & \nu\\
      -\nu & -\mu
    \end{bmatrix}},
     \qquad
     \wt{B} = SB = \sqrt{\frac{1}{\nu} R} B.
\end{equation}
Therefore, for the one-mode OQHO with the CCR matrix $\Theta$ in (\ref{Theta}) and  a positive definite energy matrix $R$ in (\ref{H}), it can be considered  that 
\begin{equation}
\label{Rscal}
    R = \nu I_2,
    \qquad
    \nu >0,
\end{equation}
without loss of generality.
In this case, the matrix $A$ reduces to
\begin{equation}
\label{Amunu}
    A
    =
    \nu \bJ  - \mu I_2
    =
    {\begin{bmatrix}
      -\mu & \nu\\
      -\nu & -\mu
    \end{bmatrix}}
\end{equation}
(in accordance with (\ref{AR}), (\ref{AR1}), (\ref{ABnew})) and its Hurwitz property is equivalent to that the coupling matrix $M$ satisfies $\mu>0$ in (\ref{MJM}). Similarly to the classical case, the quantity
\begin{equation}
\label{vartheta}
  \vartheta:= \frac{1}{\mu}
\end{equation}
describes a typical time of transient processes in the OQHO.

 \section{EIGENBASIS OF THE TWO-POINT COMMUTATOR KERNEL}
\label{sec:eig}

Regardless of a particular system-field quantum state, the system variables of the OQHO satisfy  the two-point CCRs \cite{VPJ_2018a}
\begin{equation}
\label{XXcomm}
    [X(s), X(t)^\rT]
    =
    2i\Lambda(s-t),
    \qquad
    s,t\>0,
\end{equation}
with
\begin{equation}
\label{Lambda}
    \Lambda(\tau)
     :=
     \frac{1}{2}
    \left\{
    {\begin{matrix}
    \re^{\tau A}\bJ& {\rm if}\  \tau\> 0\\
    \bJ\re^{-\tau A^{\rT}} & {\rm if}\  \tau< 0
    \end{matrix}}
    \right.
    =
    -\Lambda(-\tau)^\rT,
\end{equation}
which follows from (\ref{Xsol}) and the commutativity between the forward increments of the quantum Wiener   process and the past system variables: $[\rd W(\tau), X(s)^\rT] = 0$  for all $\tau\> s\>0$. In particular, the one-point CCRs (\ref{XCCR})--(\ref{Theta}) are recovered from (\ref{XXcomm}) as $\Lambda(0) = \Theta$.

In view of the discussion in the previous section,
it is assumed that the one-mode OQHO under consideration has a scalar  energy matrix (\ref{Rscal}) and a coupling matrix $M\in \mR^{m\x 2}$ with $\mu>0$ in (\ref{MJM}), so that the matrix $A$ in (\ref{Amunu}) is Hurwitz.  An appropriate subset of $\mR^{2\x 2}$ is isomorphic to the complex plane $\mC$ according to the correspondence
$    {\small\begin{bmatrix}
      x & -y\\
      y & x
    \end{bmatrix}}
    =
    xI_2 - y\bJ
    \leftrightarrow
    x+iy
$ for all     $    x,y\in \mR$,
whereby such matrices commute with each other.
 In particular, this commutativity holds for the matrices $
    \bJ \leftrightarrow -i$ and
$
    A \leftrightarrow -\mu -i\nu
$
 in (\ref{bJ}), (\ref{Amunu}),
so that the two-point CCR function in (\ref{Lambda}) can be represented as
\begin{equation}
\label{Lambda1}
  \Lambda(\tau)
  =
  \frac{1}{2}
  C(\tau)
  U(\tau) \bJ
  \leftrightarrow
  -\frac{i}{2} \re^{-\mu|\tau|-i\nu\tau},
\end{equation}
where
\begin{equation}
\label{C}
  C(\tau):= \re^{-\mu|\tau|},
  \qquad
  \tau \in \mR.
\end{equation}
Here, we have used the following orthogonal matrices of order $2$:
\begin{equation}
\label{U}
  U(\tau)
  :=
  \re^{\tau \nu\bJ}
  =
  {\begin{bmatrix}
    \cos(\nu\tau ) & \sin(\nu\tau)\\
    -\sin(\nu\tau ) & \cos(\nu\tau)
  \end{bmatrix}}
  =
  U(-\tau)^\rT
  \leftrightarrow
  \re^{-i\nu\tau},
\end{equation}
which (as a function of $\tau \in \mR$) form a one-parameter  group of planar rotations (with the infinitesimal generator $U'(0)=\nu \bJ$) and commute with $\bJ$ and between themselves. This gives rise to a time-varying symplectic transformation
\begin{equation}
\label{Xhat}
  \wh{X}(t)
  :=
  U(t)^\rT X(t)
  =
  U(-t)X(t),
\end{equation}
so that the transformed  system variables retain the one-point CCRs (\ref{XCCR}), while their two-point  CCRs are essentially ``scalarised'':
\begin{align}
\nonumber
    [\wh{X}(s),\wh{X}(t)^\rT]
    & =
    U(s)^\rT [X(s),X(t)^\rT] U(t)\\
\nonumber
     & = 2iU(s)^\rT \Lambda(s-t) U(t)\\
\nonumber
    & =
    i C(s-t)\bJ U(-s)U(s-t)U(t)\\
\label{XCCR2}
    & = i C(s-t)\bJ,
    \qquad
    s,t\>0,
\end{align}
where, due to (\ref{Lambda1}), the dependence on the time arguments is present only in the scalar factor $C(s-t)$ given by (\ref{C}). Now note that, in view of $\mu>0$,  the function $C$ is the invariant covariance function for an Ornstein-Uhlenbeck process \cite{KS_1991} $\zeta$ governed by a classical linear SDE $\rd \zeta = -\mu \zeta\rd t + \sqrt{2\mu}\rd \omega$, which is driven by a standard Wiener process $\omega$. For a fixed but otherwise arbitrary finite time-horizon $T>0$, the covariance kernel admits Mercer's representation
\begin{equation}
\label{Cff}
    C(s-t)
    =
    \sum_{k=1}^{+\infty}
    \lambda_k
    f_k(s)f_k(t),
    \qquad
    0\< s,t\< T,
\end{equation}
in terms of the eigenvalues $\lambda_k$
and orthonormal eigenfunctions $f_k: [0,T]\to \mR$ of a positive definite self-adjoint  operator $\cC$  which maps a square integrable function $f$ on $[0,T]$ to another such function $g:= \cC f$ as
\begin{equation}
\label{cC}
    g(s)
    :=
    \int_0^T
    C(s-t)
    f(t)
    \rd t,
    \qquad
    0\< s\< T
\end{equation}
(the Hilbert space $L^2([0,T],\mR)$ is endowed with the standard inner product
$
    \bra \varphi, \psi\ket := \int_0^T \varphi(t) \psi(t)\rd t
$ and the norm $\|\varphi\|:= \sqrt{\bra \varphi, \varphi\ket}$).
In view of (\ref{C}), (\ref{Cff}),
the eigenvalues of the operator $\cC$ in (\ref{cC})  satisfy $\sum_{k=1}^{+\infty} \lambda_k = \Tr \cC = \int_0^T C(0)\rd t = T$ and the bounds
\begin{equation}
\label{lam}
    0
    <
    \lambda_k
    \<
    \sup_{0\< s\< T}
    \int_0^T |C(s-t)|\rd t
    <
    2\int_0^{+\infty} \re^{-\mu \tau} \rd \tau
    =
    \frac{2}{\mu}.
\end{equation}
By using the
distributional \cite{V_2002} derivatives
$C'(\tau) = -\mu \sign(\tau) C(\tau)$ and $C''(\tau) = -2\mu \delta(\tau) +\mu^2C(\tau)$, with $\sign(\cdot)$ and $\delta(\cdot)$ the sign and Dirac's delta functions (or by splitting (\ref{cC}) into the sum of two integrals over the intervals $[0,s]$ and $[s,T]$), it follows that $\cC: f\mapsto g$ describes the solution of the boundary value problem
\begin{align}
\label{g''}
    g''(s) - \mu^2 g(s) & = -2\mu f(s),
    \qquad
    0\< s\< T,\\
\label{g'0}
    g'(0)
    & = \mu \int_0^T \re^{-\mu t} f(t)\rd t = \mu g(0),\\
\label{g'T}
    g'(T)
    & =
    -\mu \int_0^T \re^{\mu (t-T)} f(t)\rd t = -\mu g(T),
\end{align}
so that $C(s-t)$ in (\ref{cC}) is the corresponding Green's function.
Substitution of $g:= \lambda f$ into (\ref{cC})--(\ref{g'T}), with $0< \lambda < \frac{2}{\mu}$ in view of (\ref{lam}), represents the eigenvalue problem as the  boundary value problem
\begin{align}
\label{f''}
    f''(t) + \Big(\frac{2\mu}{\lambda}-\mu^2\Big)f(t) & = 0,
    \qquad
    0\< t\< T,\\
\label{f'0f''T}
    f'(0) & = \mu f(0),
    \qquad
    f'(T) = -\mu f(T).
\end{align}
The linear second-order ODE (\ref{f''}) leads to the eigenvalues
\begin{equation}
\label{lamom}
    \lambda_k
    =
    \frac{2\mu}{\mu^2 + \omega_k^2},
    \qquad
    k = 1,2,3,\ldots,
\end{equation}
where, in view of (\ref{f'0f''T}),  the frequencies $\omega_k>0$ form an increasing sequence of solutions of the equation
\begin{equation}
\label{trans}
    2\mu \omega_k \cos(\omega_k T) + (\mu^2-\omega_k^2) \sin(\omega_k T) = 0
\end{equation}
and specify the corresponding orthonormal eigenfunctions
\begin{equation}
\label{fk}
  f_k(t) = \frac{1}{\gamma_k}(\omega_k \cos(\omega_k t) + \mu \sin(\omega_k t)),
  \qquad
  0\< t\< T.
\end{equation}
The normalization constants $\gamma_k$ are found from the condition $\|f_k\| = 1$:
\begin{align}
\nonumber
    \gamma_k^2
     &=
    \int_0^T
    (\omega_k^2 \cos(\omega_k t)^2 + \mu^2 \sin(\omega_k t)^2 + \mu\omega_k \sin(2\omega_k t))\rd t\\
\nonumber
    & =
    \frac{T}{2}(\omega_k^2 + \mu^2)
    +
    \frac{\omega_k^2-\mu^2}{4\omega_k} \sin(2\omega_k T)
    +
    \frac{\mu}{2}(1-\cos(2\omega_k T))
    \\
\nonumber
    & =
    \frac{T}{2}(\omega_k^2 + \mu^2)
    +
    \frac{\omega_k^2-\mu^2}{2\omega_k} \sin(\omega_k T) \cos(\omega_k T)
    +
    \mu \sin(\omega_k T)^2
    \\
\label{norm}
    & =
    \frac{T}{2}(\omega_k^2 + \mu^2)
    +
    \mu
    ,
\end{align}
where
(\ref{trans}) is used. In view of (\ref{trans}),
the dimensionless quantities
\begin{equation}
\label{uk}
    u_k:= \frac{\omega_k}{\mu} = \omega_k \vartheta
\end{equation}
depend on the ratio of the time horizon $T$ and the transient time $\vartheta$ of the OQHO in (\ref{vartheta}):
\begin{equation}
\label{r}
    r:= \mu T = \frac{T}{\vartheta},
\end{equation}
so that
\begin{equation}
\label{trans1}
    2u_k \cos( r u_k) + (1-u_k^2) \sin(ru_k) = 0.
\end{equation}
Since $\frac{2u}{1-u^2} = \tan(2\arctan(u))$, then
(\ref{trans1}) is equivalent to $\sin(ru_k + 2\arctan (u_k)  ) = 0$, and the sequence of positive roots $u_1< u_2<u_3 <\ldots$ is found uniquely from
\begin{equation}
\label{pik}
    ru_k + 2\arctan (u_k) = \pi k,
    \qquad
    k = 1,2,3,\ldots,
\end{equation}
whose left-hand side is a strictly increasing function of $u_k$. The relations (\ref{lamom})--(\ref{pik}) describe the subsidiary eigenbasis associated with the two-point CCRs.

 \section{A MODIFIED QUANTUM KARHUNEN-LOEVE EXPANSION OF THE SYSTEM VARIABLES}
\label{sec:QKL}

With the transformed system variables (\ref{Xhat}),  we associate a sequence of vectors
\begin{equation}
\label{zeta}
  \zeta_k
  :=
  {\begin{bmatrix}
    \xi_k\\
    \eta_k
  \end{bmatrix}}
  :=
  \frac{1}{\sqrt{\lambda_k}}
  \int_0^T
  f_k(t)\wh{X}(t)
  \rd t,
  \qquad
  k = 1,2,3,\ldots,
\end{equation}
which consist of self-adjoint quantum variables $\xi_k$, $\eta_k$ on the system-field space  $\fH$. By applying (\ref{XCCR2}), it follows that they satisfy the CCRs
\begin{align}
\nonumber
    [\zeta_j, \zeta_k^\rT]
    & =
    \frac{1}{\sqrt{\lambda_j\lambda_k}}
    \int_{[0,T]^2}
    f_j(s)
    f_k(t)
    [\wh{X}(s), \wh{X}(t)^{\rT}]
    \rd s
    \rd t\\
\nonumber
    & =
    \frac{i}{\sqrt{\lambda_j\lambda_k}}
    \int_{[0,T]^2}
    f_j(s)
    f_k(t)
    C(s-t)
    \rd s
    \rd t
    \bJ\\
\label{zetacomm}
    & =
    \frac{i}{\sqrt{\lambda_j\lambda_k}}
    \bra
        f_j,
        \cC f_k
    \ket
    \bJ
    =
    i
    \sqrt{\frac{\lambda_k}{\lambda_j}}
    \bra
        f_j,
        f_k
    \ket
    \bJ
    =
    i
    \delta_{jk}
    \bJ,
\end{align}
where the orthonormality $\bra f_j, f_k\ket = \delta_{jk}$ of the eigenfunctions of the operator $\cC$ in (\ref{cC}) is also used, with $\delta_{jk}$ the Kronecker delta.  Therefore, $\xi_k$, $\eta_k$ form conjugate pairs of quantum mechanical positions and momenta, which commute for different $k$.  A combination of (\ref{Xhat}) with (\ref{zeta}) represents the system variables of the OQHO as
\begin{equation}
\label{XQKL}
    X(t)
    =
    U(t)
    \sum_{k=1}^{+\infty}
    \sqrt{\lambda_k}
    f_k(t)
    \zeta_k,
    \qquad
    0\< t\< T,
\end{equation}
which is a modified version of the quantum Karhunen-Loeve (QKL) expansion  \cite{VPJ_2019b}. Here, in view of (\ref{U}), (\ref{fk}), the factors $f_k(t)U(t)$ are sinusoidal functions of time.

The representation (\ref{XQKL}) is based on the two-point CCRs  of the system variables regardless of the quantum state. Now, if the OQHO is driven by vacuum input fields \cite{P_1992}, then the system variables (\ref{Xqp}) have a unique invariant multipoint Gaussian quantum state \cite{VPJ_2018a} with zero mean and the two-point quantum covariance matrix
\begin{equation}
\label{EXX}
  \bE(X(s)X(t)^\rT)
  =
  \Sigma(s-t) + i\Lambda(s-t),
  \qquad
  s,t\> 0,
\end{equation}
where $\bE \xi := \Tr(\rho \xi)$ is the quantum expectation over an underlying density operator $\rho$. Here, the imaginary part of the quantum covariances is described by (\ref{XXcomm}), (\ref{Lambda}), while the real part is given by
\begin{equation}
\label{K}
    \Sigma(\tau)
    =
    \left\{
    {\begin{matrix}
    \re^{\tau A}P& {\rm if}\  \tau\> 0\\
    P\re^{-\tau A^{\rT}} & {\rm if}\  \tau< 0
    \end{matrix}}
    \right.=
    C(\tau)
    \left\{
    {\begin{matrix}
    U(\tau)P & {\rm if}\  \tau\> 0\\
    PU(\tau) & {\rm if}\  \tau< 0
    \end{matrix}}
    \right.
    =
    \Sigma(-\tau)^\rT,
\end{equation}
where (\ref{C}), (\ref{U}) are used, and $
  P + i\Theta = P + \frac{i}{2}\bJ
$
is the invariant one-point quantum covariance matrix of the system variables in view of (\ref{Theta}). While the CCR matrix $\Theta$ is related to $A$, $B$ by the PR condition (\ref{PR}), the matrix $P = P^\rT \in \mR^{2\x 2}$ coincides with the controllability Gramian of the pair $(A,B)$ and is a unique solution of the algebraic Lyapunov equation
$
    AP + P A^\rT + BB^\rT = 0
$ which, in view of (\ref{Amunu}) and the antisymmetry of $\bJ$ in (\ref{bJ}), takes the form
\begin{equation}
\label{ALE}
    -2\mu P + \nu [\bJ, P] + BB^\rT = 0.
\end{equation}
This equation can be solved by using the Pauli matrices (\ref{Pauli}) and their Lie-algebraic properties  \cite{D_2006} together with the fact that real symmetric matrices of order $2$ form a three-dimensional subspace of $\mR^{2\x 2}$ spanned by $I_2$, $\sigma_1$, $\sigma_3$. More precisely, if $b_0, b_1, b_3 \in \mR$ are the coefficients  of $BB^\rT = b_0 I_2 + b_1 \sigma_1 + b_3\sigma_3$ over this basis, then, due to the identities $[\bJ, \sigma_1] = 2\sigma_3$ and $[\bJ, \sigma_3] = -2\sigma_1$,   the solution of (\ref{ALE}) is given by
\begin{equation}
\label{Psigma}
    P
    =
    \frac{1}{2}
    \Big(
        \frac{b_0}{\mu} I_2
        +
        \frac{1}{\mu^2 + \nu^2}
        (
            (\mu b_1  -\nu b_3)\sigma_1
            +
            (\nu b_1  +\mu b_3)\sigma_3
        )
    \Big).
\end{equation}
Since deterministic linear transformations of quantum variables in joint Gaussian states lead to Gaussian quantum variables \cite{KRP_2010}, then $\xi_1, \xi_2, \xi_3,\ldots$ and  $\eta_1, \eta_2, \eta_3,\ldots$ in (\ref{zeta}) are in  a zero-mean Gaussian quantum state with the covariances
\begin{equation}
\label{covzeta}
  \bE(\zeta_j \zeta_k^\rT)
  =
  P_{jk} + \frac{i}{2}\delta_{jk} \bJ,
  \qquad
  j,k=1,2,3,\ldots.
\end{equation}
Their real parts  $P_{jk} = P_{kj}^\rT \in \mR^{2\x 2}$ are computed by using (\ref{Xhat}), (\ref{EXX}), (\ref{K}) as
\begin{align}
\nonumber
    P_{jk}
   =&
  \frac{1}{\sqrt{\lambda_j\lambda_k}}
  \int_{[0,T]^2}
  f_j(s)f_k(t)
  \Re \bE (\wh{X}(s)\wh{X}(t)^\rT)
  \rd s\rd t\\
\nonumber
  = &
  \frac{1}{\sqrt{\lambda_j\lambda_k}}
  \int_{[0,T]^2}
  f_j(s)f_k(t)
  U(s)^\rT
  \Sigma(s-t)
  U(t)
  \rd s\rd t\\
\nonumber
  = &
  \frac{1}{\sqrt{\lambda_j\lambda_k}}
  \int_{[0,T]^2}
  f_j(s)f_k(t)
  C(s-t)
  U(s)^\rT
  (
  \chi_{s\>t}
  U(s-t)P
  +
  \chi_{s< t}
  PU(s-t)
  )
  U(t)
  \rd s\rd t\\
\nonumber
  = &
  \frac{1}{\sqrt{\lambda_j\lambda_k}}
  \int_{[0,T]^2}
  f_j(s)f_k(t)
  C(s-t)
  (
  \chi_{s\>t}
  U(t)^\rT P U(t)
  +
  \chi_{s< t}
  U(s)^\rT
  PU(s)
  )
  \rd s\rd t\\
\label{Ezz}
  = &
  \frac{1}{\sqrt{\lambda_j\lambda_k}}
  \int_{[0,T]^2}
  f_j(s)f_k(t)
  C(s-t)
  U(s\wedge t)^\rT P U(s\wedge t)
  \rd s\rd t,
\end{align}
where $\chi_\bullet$ is the indicator function of a set, and use is also made of
(\ref{C}) along with the group property of the matrix-valued function $U$ in (\ref{U}). The integration of the sinusoidal functions (with an exponential weight) on the right-hand side of (\ref{Ezz}) can be carried out in closed form but is cumbersome and is therefore omitted. The relations (\ref{Psigma})--(\ref{Ezz}) specify the statistical properties of the zero-mean Gaussian quantum variables in (\ref{zeta}) which play the role of coefficients in the modified QKL expansion (\ref{XQKL}).

\section{COMPUTING THE QUADRATIC-EXPONENTIAL FUNCTIONALS}
\label{sec:QEF}

For the one-mode OQHO, described by (\ref{XCCR}), (\ref{dX}), (\ref{AB}), (\ref{Rscal}) and (\ref{MJM}) with $\mu>0$, and assuming the time horizon $T$ to be fixed as before,
consider a quadratic exponential functional (QEF) $\Xi$ given by \cite{VPJ_2018a}:
\begin{equation}
\label{XiQ}
    \Xi
    :=
    \bE \re^{\theta Q},
    \qquad
    Q
    :=
    \int_0^T
    X(t)^{\rT} \Pi X(t)
    \rd t.
\end{equation}
Here, $Q$ is a positive semi-definite self-adjoint quantum variable,   which depends quadratically  on the past history of the system variables (over the time interval $[0,T]$) as specified by
a positive semi-definite matrix $\Pi = \Pi^\rT \in \mR^{2\x 2}$.
Also,
 $\theta>0$ is a risk-sensitivity parameter which is assumed to be sufficiently small to ensure that $\Xi<+\infty$. In view of the asymptotic behaviour
 $
    \lim_{\theta\to 0+}\frac{\ln \Xi}{\theta} = \bE Q
 $,
 the QEF $\Xi$ extends the mean square cost $\bE Q$, which is used, for example, in coherent quantum LQG control problems \cite{MP_2009,NJP_2009,SVP_2017,VP_2013a}. For finite values of $\theta>0$, the cost functional $\Xi$ imposes an exponential penalty on the past history of the system variables  captured by $Q$ in a quadratic fashion. This structure is different from (yet closely related to) the time-ordered exponentials in the original quantum risk-sensitive performance criteria for control problems  \cite{J_2004,J_2005} (see also \cite{DDJW_2006} and their subsequent development for quantum filtering problems \cite{YB_2009}).
The modified QKL expansion (\ref{XQKL})  allows the QEF in (\ref{XiQ}) to be represented as
\begin{equation}
\label{QX}
  Q
=
    \sum_{j,k=1}^{+\infty}
    \sqrt{\lambda_j\lambda_k}
    \zeta_j^\rT
    G_{jk}
    \zeta_k,
    \quad
    G_{jk}
    :=
      \int_0^T
      f_j(t)f_k(t)
      U(t)^\rT
      \Pi
      U(t)
    \rd t.
\end{equation}
The matrices $G_{jk}= G_{kj}^\rT \in \mR^{2\x 2}$ are simplified significantly
in the case of a scalar weighting matrix, when $\Pi = I_2$ without loss of generality (since a scalar factor can be absorbed by $\theta$ in (\ref{XiQ})). In this case, 
\begin{equation}
\label{G1}
    G_{jk}
    :=
      \int_0^T
      f_j(t)f_k(t)
      U(t)^\rT
      U(t)
    \rd t
    =
    \bra f_j, f_k\ket I_2
    =
    \delta_{jk}I_2
\end{equation}
due to the orthogonality of the matrix $U(t)$ in (\ref{U}) and the orthonormality of the eigenfunctions $f_k$ in  (\ref{fk}). Substitution of (\ref{G1}) into (\ref{QX}) leads to
\begin{equation}
\label{Qsplit}
    Q
    =
    \sum_{k=1}^{+\infty}
    \lambda_k
    \zeta_k^\rT
    \zeta_k
    =
    \sum_{k=1}^{+\infty}
    \lambda_k
    (\xi_k^2
    +
    \eta_k^2).
\end{equation}
The truncation of this series allows the QEF $\Xi$ in (\ref{XiQ}) to be represented as
\begin{equation}
\label{XiQN}
    \Xi
    =
    \lim_{N\to +\infty}
    \Xi_N,
    \qquad
    \Xi_N
    :=
    \bE
    \re^{\theta Q_N},
    \qquad
    Q_N
    :=
    \sum_{k=1}^{N}
    \lambda_k
    \zeta_k^\rT
    \zeta_k.
\end{equation}
Here, due to the interpair commutativity $[\zeta_j,\zeta_k^\rT] = 0$ for all $j\ne k$ in (\ref{zetacomm}),  and also \cite[Eqs. (49)--(51)]{VPJ_2018c},
the exponential admits the factorisations
\begin{equation}
\label{eQN}
    \re^{\theta Q_N}
    =
    \prod_{k=1}^N
    \re^{\theta \lambda_k (\xi_k^2+\eta_k^2)}
    =
    \prod_{k=1}^N
    \big(
    \re^{\frac{1}{2}\alpha_k \xi_k^2}
    \re^{\frac{1}{2}\beta_k \eta_k^2}
    \re^{\frac{1}{2}\alpha_k \xi_k^2}
    \big),
\end{equation}
with the positive scalars
\begin{equation}
\label{sympfactk}
    \alpha_k
    =
    \tanh (\theta\lambda_k),
    \qquad
   \beta_k
     =
    \sinh(2\theta\lambda_k)
\end{equation}
associated with the risk-sensitivity parameter $\theta$ and the eigenvalues (\ref{lamom}).
The factorizations (\ref{eQN}) can, in principle, be extended to $\re^{\theta Q}$ by using the series (\ref{Qsplit}). We will compute the ``truncated'' QEF $\Xi_N$ in (\ref{XiQN}) by applying the results of \cite{VPJ_2018c,VPJ_2019b} to the quadratic form $Q_N = Z_N^\rT(\diag_{1\< k\< N} (\lambda_k)\ox I_2)Z_N$, where the vector
$    Z_N:=
    {\scriptsize\begin{bmatrix}
      \zeta_1\\
      \vdots\\
      \zeta_N
    \end{bmatrix}}
     = [\xi_1, \eta_1, \ldots, \xi_N,\eta_N]^\rT
$
consists of $2N$ Gaussian quantum variables from (\ref{zeta})  with zero mean and the quantum covariance matrix
\begin{equation}
\label{KN}
    K_N
     :=
    \bE(Z_NZ_N^\rT)
    =
    P_N
    +
    \frac{i}{2} I_N\ox \bJ,
    \qquad
        P_N:= (P_{jk})_{1\< j,k\< N},
\end{equation}
in view of (\ref{zetacomm}), (\ref{covzeta}), (\ref{Ezz}).
Also, we will use
 auxiliary matrices
\begin{equation}
\label{PhiPsi}
      \Phi_N:=
    I_N
    \ox
    {\begin{bmatrix}
      1 & 0\\
      0 & 1\\
      1 & 0
    \end{bmatrix}},
    \qquad
    \Psi_N
    :=
    \diag_{1\< k\< N}
    (\alpha_k, \beta_k, \alpha_k),
    \qquad
    \Ups :=
    {\begin{bmatrix}
      0 & 1 & 0\\
      1 & 0 & -1\\
      0 & -1 & 0
    \end{bmatrix}},
\end{equation}
the second of which is a diagonal matrix of order $3N$ associated with (\ref{sympfactk}).

\begin{thm}
\label{th:QEF}
Suppose the risk-sensitivity parameter $\theta>0$ is small enough in the sense that
\begin{equation}
\label{mhoReL}
    r_N:=
    \br\Big(P_N\diag_{1\< k\< N}(2\alpha_k, \beta_k)\Big) < 1,
\end{equation}
where $\br(\cdot)$ is the spectral radius. Then the truncated QEF in (\ref{XiQN}) can be computed as
\begin{equation}
\label{XiNex}
    \Xi_N
  =
  \frac{1}{\sqrt{\det \Gamma_N}},
  \qquad
  \Gamma_N
  :=
  I_{3N} - \big(\Phi_N P_N\Phi_N^\rT  + \frac{i}{2} I_N\ox \Ups\big) \Psi_N
  \end{equation}
in terms of (\ref{Ezz}), (\ref{sympfactk})--(\ref{PhiPsi}).\hfill$\square$
\end{thm}
\begin{proof}
Application of \cite[Theorem~2]{VPJ_2018c} (see also \cite[Theorem~4]{VPJ_2019b}) leads to
\begin{equation}
\label{XiNex1}
    \Xi_N
    =
  \frac{1}{\sqrt{\det(I_{3N} - (\Phi_N K_N\Phi_N^\rT)^{\diam} \Psi_N )}},
\end{equation}
where a matrix $D:= (d_{jk})_{1\< j,k\< s}$ is mapped linearly to a symmetric matrix  $D^\diam$ of the same order 
which inherits its upper triangular part (including the main diagonal) from $D$, so that its entries are $(D^\diam)_{jk}
    :=
    {\small\left\{
        {\small\begin{matrix}
            d_{jk} &  {\rm if}\  j \< k\\
            d_{kj} &  {\rm if}\  j > k
        \end{matrix}}
    \right.}
$.
Now, (\ref{KN}), (\ref{PhiPsi})  imply that
$
    (\Phi_N K_N \Phi_N^\rT)^{\diam} = (\Phi_N P_N \Phi_N^\rT)^\diam + \frac{i}{2} (\Phi_N(I_N\ox \bJ)\Phi_N^\rT)^{\diam} = \Phi_N P_N \Phi_N^\rT + \frac{i}{2} I_N \ox \Ups
$, and its substitution into (\ref{XiNex1}) leads to (\ref{XiNex}). Also, (\ref{mhoReL}) employs
the property
$
    \br(\Phi_N P_N \Phi_N^\rT \Psi_N) = \br(P_N \Phi_N^\rT \Psi_N \Phi_N ) = r_N
$ for the matrices 
whose spectra differ only by zeros, together with
the diagonal matrix $\Phi_N^\rT \Psi_N \Phi_N =     \diag_{1\< k\< N}
    (2\alpha_k, \beta_k)$  of order $2N$ whose entries are given by (\ref{sympfactk}).
\end{proof}

For any $N=1,2,3,\ldots$, the matrix $\Gamma_N$ in (\ref{XiNex}) is a submatrix (of order $3N$) of  the next matrix  $\Gamma_{N+1}$. Therefore, a combination of (\ref{XiQN}) with (\ref{XiNex}) leads to a recursive representation of the QEF in the form
\begin{equation}
\label{Xiln}
    \ln \Xi = -\frac{1}{2}\Big(\ln\det \Gamma_1 + \sum_{N=1}^{+\infty}\ln \det \Gamma_{N+1\mid N}\Big),
\end{equation}
where $\Gamma_{N+1\mid N}$ denotes the Schur complement \cite{HJ_2007} of the block $\Gamma_N$ in $\Gamma_{N+1}$. The convergence of the series (\ref{Xiln}) is guaranteed if the quantity $r_N$ on the left-hand side of (\ref{mhoReL}) (which is a nondecreasing function of $N$ and depends on the risk-sensitivity parameter  $\theta$ in a nonlinear fashion through (\ref{sympfactk})) satisfies
$
    \lim_{N\to +\infty}
    r_N <1
$.

\section{CONCLUSION}
\label{sec:conc}

We have considered a modified version of the quantum Karhunen-Loeve expansion for the system variables of a one-mode OQHO
using the eigenvalues and eigenfunctions for their two-point CCRs. This eigenbasis is closely related to that for the covariance kernel of a classical Ornstein-Uhlenbeck process. The coefficients of the modified QKL expansion are organised as conjugate position-momentum pairs with interpair commutativity. Their statistical properties are more complicated and are studied for the case of the invariant multipoint Gaussian state of the OQHO driven by vacuum  input  fields. The QKL representation of the system variables has been applied to computing the QEF as a finite-horizon robust performance criterion for linear quantum stochastic systems.

\section*{ACKNOWLEDGEMENT}

IGV thanks Professor Robin L. Hudson for pointing out the reference \cite{IT_2010} and for other
useful discussions in the context of 
\cite{CH_2013,H_2018,VPJ_2019b}.


\end{document}